\documentclass[a4paper]{article}

\usepackage{a4wide}
\usepackage{xspace}
\usepackage{graphicx}
\usepackage[T1]{fontenc}
\usepackage[utf8x]{inputenc}
\usepackage{lmodern}
\usepackage{amsmath,amsthm,amssymb}

\newcommand{\Iref}[1]{(\hyperref[I#1]{I#1})\xspace}
\newcommand{\Irefall}{(\hyperref[I1]{I1--I3})\xspace}
\DeclareMathOperator{\dep}{depint}
\newcommand{\depint}[1]{\ensuremath{U_0V_0[#1]}}
\newcommand{\factor}[1]{\ensuremath{UV[#1]}}

\newcommand{\Abs}[2][]{\ensuremath{\|#2\|_{#1}}\xspace}
\newcommand{\abs}[1]{\ensuremath{|#1|}\xspace}

\usepackage[ruled]{algorithm}
\usepackage[noend]{algpseudocode}
\usepackage{enumerate,hyperref,verbatim}

\newcommand{\makeset}[2]{\ensuremath{ \{ #1 \: | \: #2 \} }}

\newcommand{\algpairc}{\algofont{PairComp}\xspace}
\newcommand{\algsolveeqlin}{\algofont{LinWordEqSat}\xspace}
\newcommand{\algblocksc}{\algofont{BlockComp}\xspace}

\newcommand{\algofont}[1]{\textnormal{\textsc \selectfont\sffamily  #1}}

\newcommand{\sol}[1]{\ensuremath{S(#1)}}
\newcommand{\solution}{\ensuremath{S}\xspace}

\newcommand{\letters}{\ensuremath{\Gamma}\xspace}
\newcommand{\variables}{\ensuremath{\mathcal X }\xspace}

\usepackage{color}
\newcommand{\NPclass}{{\sf NP}}

\newcommand{\PSPACE}{{\sf PSPACE}}

\newcommand{\poly}{{\sf {poly}}}

\newtheorem{lemma}{Lemma}
\newtheorem{theorem}{Theorem}

\providecommand{\Ocomp}{\mathcal{O}}

\newcommand{\twodots}{\mathinner{\ldotp\ldotp}}

\newcommand{\pos}{\ensuremath{\text{Pos}}}
\newcommand{\suppos}{\ensuremath{\text{Pos}_{\supseteq}}}
\newcommand{\subpos}{\ensuremath{\text{Pos}_{\subseteq}}}

\newtheorem*{claim}{Claim}

\usepackage{color}
\definecolor{myYellow}{rgb}{0.9,0.9,0}

\title{Word equations in nondeterministic linear space\footnote{This work was supported under National Science Centre, Poland project number 2014/15/B/ST6/00615.}}

\author{Artur Je\.z\\
Institute of Computer Science, University of Wroc{\l}aw,
	Poland}

\begin{document}
\maketitle

\begin{abstract}
Satisfiability of word equations is an important problem in the intersection of formal languages and algebra:
Given two sequences consisting of letters and variables
we are to decide whether there is a substitution for the variables that turns this equation
into true equality of strings.
The exact computational complexity of this problem remains unknown,
with the best lower and upper bounds being, respectively, \NPclass{} and \PSPACE.
Recently, the novel technique of recompression was applied to this problem,
simplifying the known proofs and lowering the space complexity to (nondeterministic) $\Ocomp(n \log n)$.
In this paper we show that satisfiability of word equations is in nondeterministic linear space,
thus the language of satisfiable word equations is context-sensitive,
and by the famous Immerman–Szelepcsényi theorem:
the language of unsatisfiable word equations is also context-sensitive.
We use the known recompression-based algorithm
and additionally employ Huffman coding for letters.
The proof, however, uses analysis of how the fragments of the equation
depend on each other as well as a~new strategy for nondeterministic choices
of the algorithm, which uses several new ideas to limit the space occupied by the letters.

\vspace{6pt}

\noindent\textbf{keywords} {Word equations, string unification, context-sensitive languages, space efficient computations, linear space}
\end{abstract}

\section{Introduction}
Solving \emph{word equations} was an intriguing problem since the dawn of computer science,
motivated first by its ties to Hilbert's 10th problem.
Initially it was conjectured that this problem is undecidable,
which was disproved in a seminal work of Makanin~\cite{makanin}.
At first little attention was given to computational complexity
of Makanin's algorithm and the problem itself;
these questions were reinvestigated in the '90~\cite{Jaffar90,Schulz90,Koscielski},
culminating in the {\sf EXPSPACE} implementation of Makanin's algorithm by Guti{\'e}rrez~\cite{Gutierrez98}.

The connection to compression was first observed by Plandowski~\cite{PlandowskiICALP},
who showed that a~length-minimal solution of size $N$ has a compressed representation of size $\poly(n,\log N)$,
yet the proposed algorithm still used the bound on the size of the smallest solution following from Makanin's
algorithm.
Plandowski further explored this approach~\cite{PlandowskiFOCS} and proposed a {\sf PSPACE} algorithm~\cite{PlandowskiSTOC},
which is the best computational complexity class upper bound up to date; a simpler {\sf PSPACE} solution also based on compression
was proposed by Jeż~\cite{wordequations}.
On the other hand, this problem is only known to be \NPclass-hard,
and it is conjectured that it is in {\sf NP}.

The importance of these mentioned algorithms lays also with the possibility to extend them (in nontrivial ways)
to various scenarios:
free groups~\cite{mak82,Diekertfreegroups,wordeqgroups},
representation of all solutions~\cite{PlandowskiSTOC2,wordequations,raz87},
traces~\cite{mat97lfcs,eqtraces},
graph groups~\cite{DiekertLohrey08},
terms~\cite{contextunification},
context free groups~\cite{twistedwordequations}, hyperbolic groups~\cite{hyperbolic_groups_old,hyperbolic_groups_new},
and others.

While the computational complexity of word equations remains unknown,
its exact space complexity is intriguing as well:
Makanin's algorithm uses exponential space~\cite{Gutierrez98},
Plandowski~\cite{PlandowskiSTOC} gave no explicit bound on the space usage of his algorithm,
a rough estimation is ${\sf NSPACE}(n^5)$,
the recent solution of Jeż~\cite{wordequations} yields ${\sf NSPACE}(n \log n)$.
Moreover, for $\Ocomp(1)$ variables a linear bound on space complexity was
shown~\cite{wordequations};
recall that languages recognisable in nondeterministic linear space
are exactly the context-sensitive languages.

In this paper we show that satisfiability of word equations can be tested in nondeterministic linear space
in terms of the number of bits of the input,
thus showing that the language of satisfiable word equations
is context-sensitive
(and by the famous Immerman–Szelepcsényi theorem:
the language of unsatisfiable word equations).
The employed algorithm is a (variant of) algorithm of Jeż~\cite{wordequations},
which additionally uses Huffman coding for letters in the equation.
On the other hand, the actual proof uses a different encoding of letters,
which extends the ideas used in a (much simpler) proof in case of $\Ocomp(1)$ variables~\cite[Section~5]{wordequations},
i.e.\ we encode the letters in the equation using factors of the original equations on which those letters ``depend''.
The exact notion of this ``dependence'' is defined and analysed..
The other new ingredient is a different strategy of compression:
roughly speaking, previously a strategy that minimised the length of the equation was used.
Here, a more refined strategy is used:
it simultaneously minimises the size of a particular bit encoding,
enforces that changes in the equation (during the algorithm)
are local,
and limits the amount of new letters that are introduced to the equation.

The bound holds when letters and variables in the input are encoded using an arbitrary prefix code,
in particular, the Huffman coding (so the most efficient one among the prefix codes) is allowed.

A conference version of this paper was presented at ICALP 2017.
The journal version contains omitted proofs and contains improvements on notation 
and presentation as well as fixes some minor errors.

\section{Notions}
\subsection{Word equations}
A word equation is a pair $(U, V)$, written as $U = V$,
where $U, V \in (\letters \cup \variables)^*$
and $\letters$ and $\variables$ are disjoint alphabets of \emph{letters} and \emph{variables},
both are collectively called \emph{symbols}.
By $n_X$ we denote the number of occurrences of $X$ in the (current) equation;
in the algorithm $n_X$ does not change till $X$ is removed from the equation,
in which case $n_X$ becomes $0$.
A~\emph{substitution} is a morphism
$\solution : \variables \cup \letters \to \letters'^*$,
where $\letters' \supseteq \letters$ and $\sol a = a$ for every $a \in \letters$,
a~substitution naturally extends to $(\variables \cup \letters)^*$.
A \emph{solution} of an equation $U = V$ is a substitution \solution{} such that
\sol U = \sol V;
given a solution \solution of an equation $U = V$
we call \sol U the \emph{solution word}.
We allow the solution to use letters that are not present in the equation,
this does not change the satisfiability:
all such letters can be changed to a fixed letter from \letters (or to $\epsilon$),
and the obtained substitution is still a solution.
Yet, the proofs become easier,
when we allow the usage of such letters.
The alphabet $\letters'$ is usually given implicitly:
as the set of letters used by the substitution. 
A \emph{substring} denotes a sequence of letters,
while a \emph{factor}: a sequence of letters and variables;
in both cases, usually the ones occurring in the equation.
A \emph{block} is a string $a^\ell$ with $\ell \geq 1$ that cannot be extended to the left nor to the right with $a$.

As we deal with linear-space, the encoding used by the input equation matters.
We assume only that the input is given by a fixed (uniquely decodable) coding,
i.e.\ each symbol in the input is always given by the same bitstring
and given a bitstrings there is at most one string of letters and variables that is encoded as this bitstring.
It is folklore that among such codes the Huffman code yields the smallest space consumption (counted in bits)
and moreover the Huffman coding can be efficiently computed,
also in linear space.
As we focus on space counted in bits and use encodings, by $\Abs{\alpha}$
we denote the space consumption of the encoding of $\alpha$,
the encoding shall be always clear from the context.
Furthermore, whenever we talk about space complexity, it is counted in bits.

For technical reasons we insert into the equation ending markers at the beginning and end of $U$ and $V$,
i.e.\ write them as $@ U @ , @ V @$ for some special symbol $@$.
Those markers are ignored by the algorithm,
yet they are needed for the encoding.

\subsection{Nondeterministic Linear Space}
We recall some basic facts about the nondeterministic space-bounded computation.
A~nondeterministic procedure is \emph{sound},
when given a unsatisfiable word equation $U = V$
it cannot transform it to a satisfiable one, regardless of the nondeterministic choices;
a~procedure is \emph{complete},
if given a satisfiable equation $U = V$ 
for some nondeterministic choices it returns a satisfiable equation $U' = V'$.
A composition of sound (complete) procedures is sound (complete, respectively).
It is enough that we show linear-space bound for one particular computation:
as the bound is known, we limit the space available to the algorithm and reject the computations exceeding it.
Thus we imagine our algorithm
as if it had extra knowledge that allows it to make the nondeterministic choices appropriately
and we bound the space only in the case of those appropriate
nondeterministic choices.
In particular, the subprocedures described later on
are written `as if' the algorithm knew a particular solution of the current equation.

\section{The (known) algorithm}
We use (a variant of) recompression algorithm~\cite{wordequations},
the proofs of correctness are omitted, yet they should be intuitively clear.
The algorithm conceptually applies the following two operations on \sol U and \sol V:
given a string $w$ and alphabet $\letters$
\begin{itemize}
\item the $\letters$ block compression of $w$ is a string $w'$ obtained by
replacing every block $a^\ell$ in $w$, where $a \in \letters$ and $\ell \geq 2$, with a fresh letter $a_\ell$;
\item the $(\letters_\ell, \letters_r)$ pair compression of $w$,
where $\letters_\ell, \letters_r$ is a partition of \letters,
is a string $w'$ obtained by
replacing every occurrence of a pair $ab \in \letters_\ell\letters_r$
with a fresh letter $c_{ab}$.
\end{itemize}
A \emph{fresh} letter means that it is not currently used in the equation,
nor in \letters,
yet each occurrence of a fixed $ab$ is replaced with the same letter.
The $a_\ell$ and $c_{ab}$ are just notation conventions,
the actual letters in $w'$ do not store the information
how they were obtained.
For shortness, we call $\letters$ block compression
the $\letters$ compression or block compression, when $\letters$ is clear from the context;
similar convention applies to $(\letters_\ell, \letters_r)$ pair compression,
called $(\letters_\ell, \letters_r)$ compression
or pair compression, when $(\letters_\ell, \letters_r)$
is clear from the context.
We say that a pair
$ab \in \letters_\ell \letters_r$ is \emph{covered}
by a partition $\letters_\ell, \letters_r$.

The intuition is that the algorithm aims at performing those compression operations on the solution word
and to this end it modifies the equation a bit
and then performs the compression operations on $U$ and $V$
(and conceptually also on the solution, i.e.\ on \sol X for each variable $X$). Below we describe, how it is performed on the equation.

\algblocksc: For the equation $U = V$ and the alphabet \letters{} of letters in this equation
(except the ending markers)
for each variable $X$ we first guess the first and last letter of \sol X
as well as the lengths $\ell, r$ of the longest prefix consisting only of $a$ for some $a \in \Gamma$,
called $a$-prefix, and $b$-suffix (defined similarly for some $b \in \Gamma$,) of \sol X.
Then we replace $X$ with $a^\ell X b^r$ (or $a^\ell b^r$ or $a^\ell$ when $\sol X = a^\ell b^r$ or $\sol X = a^\ell$);
this operation is called \emph{popping $a$-prefix and $b$-suffix}.
Then we perform the $\letters$-block compression on the equation
(this is well defined, as we can treat variables as symbols from outside $\letters$).

\begin{algorithm}[h]
	\caption{$\algblocksc(\letters)$}
	\label{alg:blocksc}
	\begin{algorithmic}[1]
	\Require $\letters$ is the set of letters in $U = V$
  \For{$X \in \variables$}
	\State let $a$, $b$ be the first and last letter of \sol X
	\State guess $\ell \geq 1$, $r \geq 0$ \label{guess ell}\Comment{$\sol X = a^{\ell} w b^{r}$, where $w$ does not begin with $a$ nor end with $b$}
	\State \Comment{If $\sol X = a^{\ell}$ then $r=0$}
	\State replace each $X$ in $U$ and $V$ by  $a^{\ell} X b^{r}$
	\Comment{$\sol X = a^{\ell} w b^{r}$ changes to $\sol X = w$}
	\If{$\sol X = \epsilon$} 
		remove $X$ from $U$ and $V$ \Comment{Guess}
	\EndIf
  \EndFor
	\For{each letter $a \in \letters$ and each $\ell\geq2$} \label{loop of compressions}
		\State replace every block $a^\ell$ in $U$ and $V$
		by a fresh letter $a_\ell$
	\EndFor
	\end{algorithmic}
\end{algorithm}

\algpairc:
For the alphabet \letters,
which will always be the alphabet of letters in the equation right before the block compression
(again: except the ending markers),
we partition  \letters into $\letters_\ell$ and $\letters_r$
(in a way described in Section~\ref{sec:pair compression-strategy})
and then for each variable $X$ guess whether \sol X begins 
with a letter $b \in \letters_r$ 
and if so, replace $X$ with $b X$ or $b$, when $\sol X = b$,
and then do a symmetric action for the last letter and $\letters_\ell$;
this operation is later referred to as \emph{popping letters}.
Then we perform the $(\letters_\ell, \letters_r)$ compression on the equation.

\begin{algorithm}[h]
  \caption{$\algpairc(\letters_\ell,\letters_r)$\label{alg:paircompc}}
  \begin{algorithmic}[1]
	\Require $\letters_\ell$, $\letters_r$ are disjoint, $\Gamma = \Gamma_\ell \cup \Gamma_r$
	\For{$X \in \variables$} \label{pop main loop}
		\State let $b$ be the first letter of \sol X  \label{guess first letter}\Comment{Guess}
		\If{$b \in \letters_r$} 
			\State replace each $X$ in $U$ and $V$ by  $bX$ \label{leftpop}
			\Comment{Implicitly change $\sol X = bw$ to $\sol X = w$}
			\If{$\sol X = \epsilon$} remove $X$ from $U$ and $V$ \Comment{Guess}
			\EndIf
		\EndIf
		\State let $a$ be the \ldots  \Comment{Perform a symmetric action for the last letter}
	\EndFor
	\For{$ab \in \letters_\ell\letters_r$} \label{loop lefta}
	  	\State replace each substring $ab$ in $U$ and $V$ by a fresh letter $c$
	\EndFor
  \end{algorithmic}
\end{algorithm}

\algsolveeqlin works in \emph{phases}, until an equation
with both sides of length $1$ is obtained:
in a single phase it establishes the alphabet $\letters$
of letters in the equation, performs the \letters compression
and then repeats: guess the partition of \letters to $\letters_\ell$ and $\letters_r$
and perform the $(\letters_\ell,\letters_r)$ compression,
until each pair $ab \in \letters^2$ was covered by some partition.

\begin{algorithm}[h]
	\caption{\algsolveeqlin}
	\label{alg:main}
	\begin{algorithmic}[1]
	\While{$|U| > 1$ or $|V| > 1$} 
		\State $\letters \gets$ letters in $U = V$ 
  		\State $\algblocksc(\letters)$
		\While{some pair in $\letters^2$ was not covered}
			\State partition $\letters$ to $\letters_\ell$ and $\letters_r$ \Comment{Guess}			
			\State $\algpairc(\letters_\ell,\letters_r)$
		\EndWhile
	\EndWhile
 \end{algorithmic}
\end{algorithm}

\subparagraph*{Correctness}
Given a solution \solution we say that some nondeterministic choices
\emph{correspond to} \solution, if they are done as if \algsolveeqlin
knew \solution. For instance, it guesses correctly the first letter of \sol X
or whether $\sol X = \epsilon$.
(The choice of a partition does not fall under this category.)

All of our procedures are sound and complete, furthermore
they transform the solutions in the sense described in the below Lemma~\ref{lem:sound and complete}.

\begin{lemma}[{\cite[Lemma 2.8 and Lemma 2.10]{wordequations}}]
\label{lem:sound and complete}
$\algblocksc$ is sound and complete;
to be more precise, for any solution \solution{} of an equation
$U = V$ for the nondeterministic choices corresponding to \solution the returned equation
$U' = V'$ has a solution $\solution'$ such that
$\solution'(U')$ is the \letters compression of \sol U
and $\solution'(X)$ is obtained from \sol X by removing the $a$-prefix and $b$-suffix,
where $a$ is the first letter of \sol X and $b$ the last,
and then performing the \letters compression.

When $\letters_\ell$ and $\letters_r$ are disjoint,
the $\algpairc(\letters_\ell,\letters_r)$ is sound and complete;
to be more precise, for any solution \solution{} of an equation
$U = V$ for the nondeterministic choices corresponding to \solution the returned equation
$U' = V'$ has a solution $\solution'$ such that
$\solution'(U')$ is the $(\letters_\ell,\letters_r)$ compression of \sol U
and $\solution'(X)$ is obtained from \sol X by removing the first letter of \sol  X,
if it is in $\letters_r$, and the last, if it is in $\letters_\ell$,
and then performing the $(\letters_\ell,\letters_r)$ compression.
\end{lemma}

The solution $\solution'$ from Lemma~\ref{lem:sound and complete}
is called a \emph{solution corresponding to} \solution after
$(\letters_\ell,\letters_r)$ compression (\letters compression, respectively);
we also talk about a \emph{solution corresponding to }\solution,
when the compression operation is clear from the context
and extend this notion to a~solution corresponding to \solution after a phase.
What is important later on is how $\solution'$ is obtained from \solution:
it is modified as if the subprocedures knew first/last letter of \sol X
and popped appropriate letters from the variables
and then compressed pairs/blocks in substitution for variables.
The nondeterministic choices in the compression operations
that yield the corresponding solution $\solution'$ are called
\emph{corresponding} as well and they are intuitively clear:
all of them are described in Algorithms~\ref{alg:blocksc}--\ref{alg:main} `as if' the solution
was given explicitly (for instance, we replace $X$ with $a^\ell X b^r$ when \sol X has an $a$-prefix $a^\ell$ and $b$-suffix $b^r$).

Lemma~\ref{lem:sound and complete} yields the soundness and completeness of \algsolveeqlin,
for the termination we observe that iterating the compression operations shortens the string by a~constant fraction,
thus the length of a solution word shortens by 
a constant fraction in each phase.

\begin{lemma}
\label{lem:shortening the strings}
Let $w$ be a string over an alphabet $\letters$
and $w'$ a string obtained from $w$ by a~\letters compression
followed by a sequence of $(\letters_\ell,\letters_r)$ compressions
(where $\letters_\ell,\letters_r$ is a partition of \letters)
such that each pair $ab \in \letters^2$ is covered by some partition.
Then $|w'| \leq \frac{2|w|+1}{3}$.
\end{lemma}

\begin{proof}
Consider two consecutive letters $a$, $b$ in $w$.
At least one of those letters is compressed during the procedure:
\begin{itemize}
	\item \emph{if $a = b$:} In this case they are compressed
	during the \letters compression.
	\item \emph{$a \neq b$:}
	At some point the pair $ab$ is covered by some $(\letters_\ell,\letters_r)$ compression.
	If any of letters $a,b$ was already compressed then we are done.
	Otherwise, this occurrence of $ab$ is now compressed.
\end{itemize}

Hence each uncompressed letter in $w$ (except perhaps the last letter) can be associated with the two letters
to the right that are compressed.
This means that (in a phase) at least $\frac{2}{3}(|w|-1)$ letters are compressed
and so $\abs{w'} \leq \abs{w} - \frac{1}{3}(|w|-1)$, as claimed.
\end{proof}

\begin{theorem}
	\label{thm:linspace}
\algsolveeqlin is sound, complete and terminates (for appropriate nondeterministic choices) for satisfiable equations.
It runs in linear (bit) space.
\end{theorem}
The proof is given in Section~\ref{subsec:wrap-up}.

In the following, we will also need one more technical property of block compression.

\begin{lemma}
\label{lem:letters are different}
Consider a solution \solution during a phase with nondeterministic choices corresponding to \solution{} and the corresponding solution
$\solution'$ of $U' = V'$ after the block compression.
Then $\solution'(U')$ has no two consecutive letters $aa \in \letters$.
\end{lemma}
This is true after block compression and afterwards no
letters from \letters are introduced.

\subparagraph*{Compressing blocks in small space}
Storing, even in a concise way, the lengths of popped prefixes and suffixes
in \letters compression makes attaining the linear space difficult.
This was already observed~\cite{wordequations} and a linear-space
implementation of \algblocksc is known~\cite{wordequations}.
It performs a different set of operations,
yet the effect is the same as for \algblocksc.
Instead of explicitly naming the lengths of blocks,
we treat them as integer parameters;
then we declare, which maximal blocks are of the same length
(those lengths depend linearly on the parameters);
verifying the validity of such a guess is done by writing a system
of (linear) Diophantine equations that formalise those equalities
and checking its satisfiability.
This procedure is described in detail in \cite[Section~4]{wordequations}.
In the end, it can be implemented in linear bitspace.

\begin{lemma}[{\cite[Lemma 4.7]{wordequations}}]
\label{lem:block compression linear space}
\algblocksc{} can be implemented in space linear in the bit-size of the equation
\end{lemma}

\subparagraph*{Encoding the equation}
At each step of the algorithm we encode letters
(though not variables) in the equation using Huffman coding.
The variables are encoded using the original coding from the input equation.
To distinguish between the two codes,
we prefix each code for a letter with \texttt{0}
and for a variable with \texttt{1};
in this way we can also distinguish between the letters and variables.
This increases the space usage at most twice
when compared to the usage without this extra bit,
so we disregard it, as we aim for linear space (with an arbitrary constant).
Also, the space usage for variables is at most $2\Abs[0]{U_0V_0}$,
so we disregard it later on and focus on the space usage of the letters.

Using the Huffman coding of the letters
may mean that when going from $U = V$ to $U' = V'$
the encoding of letters changes and in fact using the former encoding in the latter equation may lead to super-linear space
(imagine that we pop from each variable a~letter that has a very long code).
Using standard methods changing the encoding 
during a transition from $U = V$ to $U' = V'$ 
can be done in bit-space linear in the bit-size of the input plus bit-size of the output, using perhaps different encodings.

\begin{lemma}
\label{lem:performing subprocedures in linear space}
Given a string (encoded using some uniquely decodable code),
its Huffman coding can be computed in linear bitspace.

Each subprocedure of \algsolveeqlin{} that transforms an equation
$U = V$ to $U' = V'$ can be implemented in 
bit-space $\Ocomp(\Abs[1]{U = V} + \Abs[2]{U' = V'})$,
where $\Abs[1]{\cdot}$ and $\Abs[2]{\cdot}$
are the Huffman codings for letters in $U = V$ and $U' =V'$, respectively.
\end{lemma}
\begin{proof}
In the following we will use some `fresh symbols'. To guarantee that they have short codes,
we prefix each code used in the string by, say, \texttt{0} and each fresh symbol by \texttt{0}.
This will keep the space linear.
In the end those extra bits are removed.

A standard implementation of the Huffman coding firstly calculates for each symbol in the string the number of its occurrences,
this can be done in linear space, as a symbol plus number of its occurrences takes space linear in the space taken by all those occurrences.
Then it iteratively builds an edge-labelled tree with leaves corresponding to original letters.
The labels on the path from the root to a leaf $a$ give a code for $a$.
The algorithm takes two letters with the smallest number of occurrences,
creates a new node (which is treated further on as a leaf),
attaches the two nodes to the new node and labels the edges with \texttt{0} and \texttt{1}.
This is iterated till one node is obtained.
The tree uses linear space in total and otherwise the used space only decreases.
It is easy to see that the whole computation can be done in linear space.

For popping letters in $(\letters_\ell,\letters_r)$ compression
we use fresh symbols $X\#0$, $X\#1$ that are bit-encoded as $X$ plus $\Ocomp(1)$ bits for letters popped to the left and right
and then simply list the letters that are equal.
In this way we can compute the Huffman coding after popping and translate to the new encoding.
For compression itself, when $ab$ is compressed, we encode it as $(ab)$,
where `$($' and `$)$' are fresh symbols (though the same are used for each such encoding)
and then recompute the Huffman coding.

For block compression, Lemma~\ref{lem:block compression linear space}
already states that it can be performed in space linear in the bitsize of the old equation,
the Huffman coding of the new one can be then computed.
\end{proof}

\section{Space consumption}
\label{sec:space}
In order to bound the space consumption, we will use bit-encoding of letters that depends on the current equation.
We use the term `encoding' even though it may assign different codes to different occurrences of the same letter,
but two different letters never have the same code.
Since we are interested in linear space only,
we do not care about the multiplicative $\Ocomp(1)$ factors in the space
consumption and can assume that our code is a prefix code (i.e.\ no codeword is a prefix of another),
say by terminating each encoding with a special symbol $\$$.
We show that such an encoding uses linear space,
which also shows that the Huffman encoding of the letters in the equation
uses linear space,
as Huffman code uses the smallest space among the prefix codes (which is a folklore result by now).

The idea of our `encoding' is:
for each letter in the current equation
we establish an interval $I$ of indices in the original equation
(viewed as as string $\depint{1\twodots \abs{U_0V_0}}$)
on which it `depends' (this has to be formalised)
and encode this letter as $\depint{I}\#i$,
when it is $i$th in the sequence of letters assigned $I$
and $\depint{I}$ is the factor of the original equation restricted to indices in $I$.
We prove that letters given the same encoding are indeed the same.
For the space bound, we separately count the space used by all $\depint{I}$s
and separately the one used for numbers.
We show that using a `random' partition each time we can guarantee that `on average' a single symbol $\alpha$ from the original equation
is used in constant number of $I$s
(as $\alpha$ may have different encoding, we also need to weight this using the $\Abs{\alpha}$).
For the numbers, the argument is similar: for a given $I$ the average amount
of different letters encoded as $I\#i$ is $\Ocomp(1)$
and so the total space consumption is linear (the actual calculations are a bit more tricky here).
The dependency is formalised in Section~\ref{subsec:depints},
while Section~\ref{sec:pair compression-strategy} first gives the high-level intuition
and then upper-bound on  the used space.

\subsection{Dependency intervals}
\label{subsec:depints}
The input equation is denoted by $U_0 = V_0$,
the $U=V$ and $U' = V'$ are used for the current equation and equation
after performing some operation.
We treat the input equation as a single string $U_0V_0$ and consider its
\emph{indices}, i.e.\ numbers from $1$ to $\abs{U_0V_0}$,
denoted by letterss $i, i', j$ and intervals of such indices,
denoted by letter $I, I'$ or $[i \twodots j]$.
The $\depint{I}$ and $\depint{i\twodots j}$ denotes the substring of $U_0V_0$
restricted to indices in $I$ or in $[i \twodots j]$.
We use a partial order $\leq$ on intervals:
$[i\twodots j] \leq [i'\twodots j']$ if $i \leq i'$ and $j \leq j'$.
Note that this is \emph{not} the inclusion as sets, which we will also use.

In the current equation, i.e.\ the one stored by \algsolveeqlin,
we do not consider indices but rather \emph{positions}
and denote them by letters $p,q$. We do not think of them as numbers
but rather as pointers: when $U = V$ is transformed by some operation
to $U' = V'$ but the letter/variable at position $p$ was not affected
by this transformation,
we still say that this letter/variable is at position $p$.
On the other hand, the affected letters are on positions
that were not present in $U = V$.
In the same spirit we denote by $p$ the positions in $U = V$
and the corresponding position in $\sol U = \sol V$.
We still use the left-to-right ordering on positions, use $p-1$ and $p+1$ to denote the previous and next position;
we also consider intervals of positions,
yet they are used rarely so that they are not confused with intervals of indices, on which we focus mostly.
Given an equation $U = V$ and an interval of positions $P$ by
$UV[P]$ we denote the factor of letters and variables at positions in $P$,
again, this notation is used rarely.
In the input equation the index and position is the same.

With each position $p$ in the (current) equations (including the endmarkers)
we associate \emph{dependency interval} $\dep(p)$, called \emph{depint};
if the depint is a single index $\{i\}$, we denote it $i$.
The idea is that the letter at position $p$ is uniquely determined by $\depint{\dep(p)}$
(and the nondeterministic choices of the algorithm, i.e.\ it is determined at this moment of the algorithm's run),
note that $\depint{\dep(p)}$ may include both variables and letters.
We use the notions of $\subseteq$ and $\supseteq$ for the depints
with a usual meaning; we take unions of them,
denoted by $\cup$, but only when the result is an interval.
We say that $I$ and $I'$ are similar, denoted as $I \sim I$,
if $\depint{I} = \depint{I'}$.

To get some rough intuition:
if we have several occurrences of $XwY$ in the equation,
then as long as $X$ and $Y$ are not removed, the word between them changes in the same way
and so the letters between $X$ and $Y$ have depints that are factors of $XwY$.
When $Y$ is removed, we have to resort to the next variable, say $Z$ and gradually enlarge the depint to $X w Y v Z$.

Given an interval $I$ of indices in $U_0V_0$
by $\pos(I)$ we denote positions in the current equation
whose depint is $I$, i.e.\ $\pos(I) = \makeset{p}{\dep(p) = I}$.
In the analysis it is also  convenient to look at positions
whose depint is a superset of $I$:
$\suppos(I) = \makeset{p}{  \dep(p) \supseteq I}$,
this is usually used for $I = \{i\}$.
The mnemonics is that \emph{pos}itions whose depint $\supseteq I$.

We shall ensure the following properties
\begin{enumerate}[({I}1)]
\item Given a depint $I$, the $\pos(I)$ is a (perhaps empty) interval of positions,
similarly $\suppos(I)$.\label{I1}
\item Given depints $I, I'$ such that $\pos(I) \neq \emptyset \neq \pos(I')$ either $I \leq I'$ or $I \geq I'$. \label{I2}
\item If $I \sim I'$ then the words induced by the intervals of positions $\pos(I)$ and $\pos(I')$ are the same, i.e.\ $\factor{\pos(I)} = \factor{\pos(I')}$. \label{I3}
\end{enumerate}

\subparagraph*{Encoding of letters}
Letters in $\pos(I)$ are encoded as
$\depint{I}\#1$, $\depint{I}\#2$, etc.
Note, that there is no a priori bound on the size of such numbers.
Furthermore, if $I' \sim I$ then encoding $I\#i$ and $I'\#i$ is the same
(these are the same symbols by~\Iref{3}).

\subparagraph*{Assigning depints to letters}
The original positions in the equation have depints equal to themselves,
i.e.\ position $p$ in $U_0V_0$ we have $\dep(p) = \{p\}$.
Note that the ending markers also have their depints.
When $X$ at position $p$ pops a~letter into position $p'$
then $\dep(p') \gets \dep(p)$
(which is the position of this occurrence of $X$ in the input equation).
Before we perform the $(\letters_\ell,\letters_r)$ compression
then in parallel for each position $p$ such that $UV[p] \in \letters_\ell$
we assign $\dep(p) \gets \dep(p) \cup \dep(p+1)$
($p+1$ may be a position of a variable or of an endmarker),
we say that we increase the depint in this case.
Then we perform a symmetric action for positions whose letters
are in $\letters_r$ (so for $p-1$).
A simple  argument, see Lemma~\ref{lem:extending depints},
shows that the order of operation ($\letters_\ell$ or $\letters_r$ first) does not matter.
Observe that we compress letters with the same depint
and the new position has exactly this depint,
which is the union of depints of positions from before increase of depints.
This is intuitively clear: if $ab$ is compressed to $c$ then this $c$ depends on
the union of intervals on which $a$ and $b$ depended.
However, note that if we do not compress a letter from $\Gamma_\ell$ on position $p$
then its depint should also be extended:
it was not compressed because of the letter on position $p+1$.

\begin{lemma}
\label{lem:extending depints}
The depints assigned before pair compression are the same, regardless of whether the $\Gamma_\ell$ or $\Gamma_r$
letters are considered first.

If $\factor{p\twodots p+1} \in \Gamma_\ell \Gamma_r$
then right before the compression they have the same depint.
\end{lemma}
\begin{proof}
It is enough to show that for three consecutive letters $abc$ the depint of $b$
is going to be the same, regardless of whether we consider $\letters_\ell$
or $\letters_r$ first
(note that due to endmarkers $b$ always has a symbol to the right and left).
If $b \notin \letters_\ell \cup \letters_r$ then there is nothing to prove,
as the depint remains the same;
the case $b \in \letters_\ell$ and $b \in \letters_r$ are symmetric
(note that it is always true that $\letters_\ell \cap \letters_r = \emptyset$),
so we consider only the former.

Let the depints of $b, c$ (or rather: their positions) be $I, I'$.
If we consider $\letters_\ell$ first, then in the first step $b$ gets the depint $I \cup I'$ and in the second step nothing changes.
If we consider first $\letters_r$ and $c \notin \letters_r$ then after the first step
the depints of $b, c$ are still $I, I'$ and in the second step $b$ gets depint $I \cup I'$.
If $c \in \letters_r$ then in the first step it gets the depint $I \cup I'$
and $b$ still has depint $I$.
Then in the second step $b$ gets depint $I \cup (I \cup I') = I \cup I'$.
\end{proof}

For \letters compression, we perform in parallel the following operation for each block (perhaps of length $1$) of a letter in \letters:
given a block $a^\ell$ (for $\ell \geq 1$) at positions $p, p+1, \ldots, p+\ell-1$
we set the depints of those positions to
$\bigcup_{i=-1}^\ell \dep(p+i)$ (note that $p-1$ and $p+\ell$ are included).
Observe that:
\begin{itemize}
\item when we compress then all letters within the compressed block have the same depint;
\item we increase the depint even if $\ell = 1$, in which case there is no compression.
\end{itemize}
The intuition for including the depints of letters to the right and left of the block
is that the block depends on them as well, as they show where the block terminates.

Observe that by the way the depints are extended,
the depints of the endmarkers are never changed.

In the following we mostly focus on $\suppos(i)$.
As this is an interval of positions,
we visualize that $\suppos(I)$ extends to the neighbouring positions.
Thus we will refer to operations of changing the depints before the block compression
and pair compression as (except popping letters from variables)
in which letters become $\suppos(I)$ letters
\emph{extending} $\suppos(I)$ to new positions.
Note that this notion does not apply to the case when we pop letters from variables.
Note that the same operation may extend $\suppos(I)$ and $\suppos(I')$ to the same position.

Depints defined in this way satisfy the conditions \Irefall,
this is shown in the below Lemma~\ref{lem:basic depint occurrences}.

\begin{lemma}
\label{lem:basic depint occurrences}
\Irefall hold during \algsolveeqlin.
\end{lemma}
\begin{proof}
We first show \Iref{1} for $\suppos(i)$.
The proof is by induction;
this is true at the beginning.
If we make a union of depints, a position adjacent to a position in $\suppos(i)$
can become part of $\suppos(i)$
(this can be iterated when the depints are changed before the blocks compression),
which is fine.
During the compression, we compress symbols on positions with the same depints, so this is fine.
When we pop a letter from variable at position $p$ to position $p'$
then $\dep(p') = \dep(p) \in \suppos(i)$
and by inductive assumption $\suppos(i)$ was an interval,
so either we insert a position into it, so it is still an interval,
or we create a new position to the left or right of it,
so $\suppos(i)$ is still an interval.

Now \Iref{1} for $\suppos([i\twodots j])$ for an arbitrary
depint $[i \twodots j]$ follows:
$\suppos([i \twodots j]) = \bigcap_{k=i}^j \suppos(k)$ and as each $\suppos(k)$
is an interval, also $\suppos([i \twodots j])$ is.

We now show by induction that $i \leq i'$
implies $\suppos(i) \leq \suppos(i')$.
Clearly this holds at the beginning,
as then $\suppos(i) = \pos(i) = \{i\}$ and $\suppos(i') = \pos(i') = \{i'\}$.
Consider the moment, in which the condition $\suppos(i) \leq \suppos(i')$
is first violated,
by symmetry it is enough to consider the case in which
the first position in $\suppos(i')$ is smaller than the first in $\suppos(i)$.
If this position was just popped then it cannot be popped to the right,
as the position of popping variable is in $\suppos(i')$.
So it was popped to the left.
But then the variable that popped it was on position $p' \in \suppos(i')$ and by induction assumption $\suppos(i') \geq \suppos(i)$,
so either $p' \in \suppos(i)$, which cannot happen,
as $\depint{p}$ has only one index (this is a variable)
or there is $p \in \suppos(i)$ such that $p' > p$.
Then the new popped position is not to the left ot $p$, contradiction.
The other option is that this happened when a depint of a position $p'$
was changed so that it got into $\suppos(i')$.
But then $p'+\ell$ for some $\ell \neq 0$ was in $\suppos(i')$ and $p'$
is increased by $\dep(p' + \ell)$.
By induction assumption there was $p \in \suppos(i)$ such that $p \leq p' + \ell$.
If $p \leq p'$ then we are done, in particular, if $\ell <0$ then we are done,
so in the following we consider $\ell > 0$.
Then $p' < p \leq p'+\ell$ and so by the way the depints are increased,
$\dep(p')$ is increased as well by $\dep(p)$, so $p' \in \suppos(i)$.
As $p'$ is left-most in $\suppos(i')$, we have that left-most position in $\suppos(i)$ is less or equal to $p'$.

Concerning \Iref{2}, we show that for $p < p'$ implies $\dep(p) \leq \dep(p')$.
Let $\dep(p) = [i\twodots j]$ and $\dep(p') = [i'\twodots j']$.
If $\dep(p) \leq \dep(p')$ does not hold then either $i > i'$ or $j > j'$.
We consider the former, the proof for the latter is symmetric.
In particular, $i' \notin \dep(p)$ and so $p \notin \suppos(i')$.
We already showed that then $\suppos(i) \geq \suppos(i')$.
So if $p+1 \in \suppos(i') \leq \suppos (i) \ni p $
then also $p \in \suppos(i')$, contradiction.

For the purpose of the proof of \Iref{3}, define
$\subpos(I) = \makeset{p}{\dep(p) \subseteq I}$
(a dual notion to $\suppos(I)$).

\begin{claim}
\label{aux_claim}
$\subpos(I)$ is an interval of positions.
Given two similar depints $I \sim I'$ it holds that
$UV[\subpos(I)] = UV[\subpos(I')]$
and the corresponding positions in them have similar depints.
\end{claim}

Note that \Iref{3} follows from the Claim:
clearly $\pos(I) \subseteq \subpos(I)$ (and $\pos(I') \subseteq \subpos(I')$),
by the Claim $UV[\subpos(I)] = UV[\subpos(I')]$
and the corresponding positions in them have similar depints,
in particular $\pos(I)$ and $\pos(I')$ are positions with depints $I$ and $I'$,
respectively,
so $UV[\pos(I)] = UV[\pos(I')]$.
Moreover, the Claim also implies \Iref{1} for $\pos(I)$,
as $\pos(I) = \suppos(I) \cap \subpos(I)$ and both are intervals,
so also $\pos(I)$ is.

It remains to show the Claim,
we do it by induction on the number of operations performed by the algorithm.
At the beginning we have $\subpos(I) = \pos(I) = I$ and similarly
$\subpos(I') = \pos(I') = I'$.
Then $I \sim I'$ by definition means that $\depint{I} = \depint{I'}$
and as at the beginning $UV = U_0V_0$, we get the claim.

If a variable $X$ at position $ p \in \subpos(I)$ pops a letter,
by inductive assumption $X$ occurs at the corresponding position $p' \in \subpos(I')$
and by the algorithm it pops the same letters (to the same side)
and the positions of those letters have depints $\dep(p)$ and $\dep(p')$ respectively, and $\dep(p) \sim \dep(p')$.
So those new positions are included in $\subpos(I)$ and $\subpos(I')$,
respectively, and they are at corresponding places.
In particular, as $p, p'$ are next to each other,
$\subpos(I)$ is still an interval.

If letters at positions $p, \ldots p + \ell$ are compressed
then right before the compression those positions have the same depints
and in $\subpos(I')$ on corresponding positions $p',\ldots, p'+\ell$
there are the same letters with similar depints;
in particular if one of those letters is in $\subpos(I)$ ($\subpos(I')$)
then all of them are,
so the corresponding letters are compressed in the same way,
also afterwards the resulting letter is still within $\subpos(I)$ ($\subpos(I')$),
as it has the same depint as the compressed letters.
Moreover, as we replace consecutive positions that have the same depint
with a position of this depint, each $\subpos(I)$ remains an interval.

The last possibility is that the depint of some position is increased,
say for $p \in \subpos(I)$ we assign $\dep(p) \gets \bigcup_{h=-k}^\ell \dep(p+h)$;
this corresponds to either the block compression or pair compression.
We first show that each $\subpos(I)$ is still an interval.
\begin{itemize}
\item In the pair compression note that $\dep(p) \gets \dep(p) \cup \dep(p')$ where $p' = p+1$ or $p' = p-1$.
If $p$ seizes to be in $\subpos(I)$ then $p' \notin \subpos(I)$ and so $\subpos(I)$ looses its first or last position,
so it is still an interval.

\item In the block compression each position in $[p-k+1 \twodots p+\ell-1]$ gets the same depint.
So either all of them are in $\subpos(I)$ or none.
If none then some of them was not in $\subpos(I)$ before and as $\subpos(I)$ is an interval,
we now remove its prefix or suffix, so it is still an interval.
\end{itemize}
Let $p' \in \subpos(I')$ be the corresponding position in $\subpos(I')$.
Then the letter at position $p'$ is the same as the one at $p$
and so its depint is replaced with
$\dep(p) \gets \bigcup_{h=-k'}^{\ell'} \dep(p+h)$.
We show that if $[p-k\twodots p+\ell] \subseteq \subpos(I)$  then
$k = k'$ and $\ell = \ell'$ and
$[p'-k\twodots p'+\ell] \subseteq \subpos(I')$:
\begin{itemize}
\item For pair compression, as the letters at position $p$ and $p'$ are the same by the induction assumption,
then $k, \ell$ and $k', \ell'$ are determined by this letter and so $k = k'$ an $\ell = \ell'$.
\item For the block compression, let the letter at position $p$ be $a$.
Then the $a$ block is on positions $[p-k+1\twodots p + \ell-1]$,
in particular the letters on positions $p-k$ and $p+\ell$ are not $a$.
Then by induction assumption the corresponding letters at positions $[p'-\ell \twodots p'+k]$
are the same, in particular the $a$-block is of the same length.
Hence $k = k'$ and $\ell = \ell'$.
\end{itemize}
From the induction assumption in this case the corresponding letters are the same
and have similar depints,
so the new depints of $p, p'$ are similar.
Finally, by symmetry, if 
$[p'-k'\twodots p'+\ell'] \subseteq \subpos(I')$,
then
$[p-k\twodots p+\ell] \subseteq \subpos(I)$,
$k = k'$ and $\ell = \ell'$, etc.
So the claim holds if at least one of the positions $p, p'$ is in $\subpos(I)$ or
$\subpos(I')$ after the increase of depints.
It remains to show that if $p$ ($p'$) is removed from $\subpos(I)$
($\subpos(I')$, respectively),
then also $p'$ is removed from $\subpos(I')$ ($p$ from $\subpos(I)$, respectively).
We already showed that $[p'-k'\twodots p'+\ell'] \subseteq \subpos(I')$ if and only if 
$[p-k\twodots p+\ell] \subseteq \subpos(I)$.
By negating both sides of this statement we obtain that
$[p-k\twodots p+\ell] \not \subseteq \subpos(I)$ 
if and only if
$[p'-k'\twodots p'+\ell'] \not \subseteq \subpos(I')$.
The first condition means that
$\bigcup_{h = -k}^\ell \dep(p+h) \not \subseteq I$
and the latter that
$\bigcup_{h = -k'}^{\ell'} \dep(p'+h) \not \subseteq I'$.
So $p$ is removed from $\subpos(I)$ if and only if $p'$ is removed from $\subpos(I')$.
\end{proof}

\subsection{Pair compression strategy}
\label{sec:pair compression-strategy}
We assume that \algsolveeqlin{} makes the nondeterministic choices
according to the solution,
thus the space consumption of a particular non-deterministic execution
depends only on the choices of the partitions during pair compression, called \emph{a strategy}.
We describe a strategy leading to a linear-space sage.

\subparagraph*{Idea}
\label{subsubsec:idea}
Imagine we ensured that during one phase each variable popped $\Ocomp(1)$
letters and each $\suppos(i)$ expanded by $\Ocomp(1)$ letters.
Then $\abs{\suppos(i)} = \Ocomp(1)$:
we introduced $\Ocomp(1)$ positions to $\pos(i)$,
say at most $k$, and by Lemma~\ref{lem:shortening the strings}
among positions in $\suppos(i)$ at the beginning of the phase
at least $2/3$ took part in compression,
so their number dropped by $1/3$;
thus $\abs{\suppos(i)} \leq 3k$.
As a result, $\abs{\pos(I)}\leq 3k$ for each depint $I$:
as $\pos(I) \subseteq \suppos(i)$ for $i \in I$.
This would yield that the whole bit-space used for the encoding is linear:
each number $m$ used in $\depint{I}\#m$ is at most $3k = \Ocomp(1)$,
so they increase the size by at most a constant fraction.
On the other hand, the depints consume:
\begin{align*}
\sum_{I: \text{depint}} \Abs{\depint{I}} \cdot \abs{\pos(I)}
=
\sum_{i: \text{index}} \Abs{\depint{i}} \cdot \abs{\suppos(i)} 
\end{align*}
(a simple proof is given later)
and the right hand side is linear in terms of the input equation:
$\abs{\suppos(i)} = \Ocomp(1)$ and $\sum_{i: \text{index}} \Abs{\depint{i}}$
is the the bit-size of the input equation.

Unfortunately, we cannot ensure that each variable
pops $\Ocomp(1)$ letters nor that each $\suppos(i)$ extends by $\Ocomp(1)$ positions.
We can make this true \emph{in expectation}:
Given a phase, we call a letter \emph{new},
if it was introduced during this phase.
New letters cannot be popped nor can $\suppos(i)$ be extended by
positions with new letters.
Thus they are used to prevent extending $\suppos(i)$ and popping:
it is enough to ensure that the first/last letter of a variable
is new and that a letter on the position to the left/right of $\suppos(I)$ is new.
Now, given a random partition there is a $1/4$
probability that a fixed pair is compressed (and the resulting letter is new),
which in a sense means that in expectation a new letter will appear within $\Ocomp(1)$ letters from each end of $\sol X$
and within $\Ocomp(1)$ letters to the left and right of $\suppos(I)$.
It remains to formalise this approach and show that the ``expectation'' translates
to only $\Ocomp(1)$ times worse worst-case performance.

\subparagraph*{Strategy}
Given a solution \solution{} of an equation we say that a variable
$X$ is \emph{left blocked} if \sol X has at most one letter or the first or second letter in \sol X is new, otherwise a variable is \emph{left unblocked};
define \emph{right blocked} and \emph{right unblocked} variables similarly.
Let an index $i$ be such that $\depint{i}$ is not a variable,
then $i$ is \emph{left blocked} if in \sol U (or \sol V, respectively)
there is at most two positions to the left of $\suppos(i)$ (and none of them is a position of a variable)
or one of the letters on the positions one and two to the left of $\suppos(i)$ is new,
otherwise $i$ is \emph{left unblocked}; 
define \emph{right blocked} and \emph{right unblocked} indices similarly.

Note that $\suppos(i)$ contains only letters from the equation (as well as variables),
so letters in $\sol U$ that were obtained from substitution for a variable
are not included in it.

\begin{lemma}
\label{lem:blocking}
Consider a solution $\solution = \solution_0$ and consecutive solutions $\solution_1, \solution_2, \ldots$
corresponding to it during a phase.

If a variable $X$ becomes left (right) blocked for some $\solution_k$,
then it is left (right, respectively) blocked for each $\solution_\ell$ for $\ell \geq k$
and it pops to the left (right, respectively) at most $1$ letter after it became left (right, respectively) blocked.

If an index $i$ becomes left (right) blocked for some $\solution_k$
then it is left (right, respectively) blocked for each $\solution_\ell$ for $\ell \geq k$ 
and at most one letter to the left (right, respectively)
will have its depint extended by $i$ after $i$ became left (right, respectively) blocked.
\end{lemma}

\begin{proof}
If $X$ becomes left blocked because $\solution_k(X)$ has one letter,
then it will stay left blocked and can pop at most one letter further on.
If it becomes left blocked because first or second letter of $\solution_k(X)$ is new
then this new letter cannot be popped, as we pop only letters from \letters,
so this letter will remain on first or second position within
$\solution_\ell(X)$ for $\ell \geq k$ (so in this phase)
and it keeps $X$ left blocked.
In particular, if this letter
is first (second) in $\solution_k(X)$, then $X$ cannot pop left a letter
(can pop at most one letter); a similar argument applies on the right side.

Similarly, only positions with letters from \letters
(this does not include endmarkers)
can have their depints increased,
let $p$ be the leftmost position in $\suppos(i)$.
If a letter at position $p-2$ or $p-1$ is new,
then $\suppos(i)$ cannot extend to this position and so it will extend by at most one position
and will remain left-blocked.
Similarly, when there is only one letter (or no letter) to the left of $\suppos(i)$
then $\suppos(i)$ can extend only by this letter and it will remain left blocked.
A symmetric argument applies for right blocked depints. 
\end{proof}

The strategy iterates steps \ref{eq:variable weight}--\ref{eq:depint count}.
In a step $i$ it chooses a partition so that the corresponding $i$-th
sum decreases by a maximum amount (we show that this is at least half),
unless this sum is already $0$:

\begin{align}
\label{eq:variable weight}
\sum_{\substack{X \in \variables\text{ left unblocked}}} n_X \cdot \Abs{X} + \sum_{\substack{X \in \variables \text{ right unblocked}}} n_X \cdot \Abs{X}\\
\label{eq:depint weight}
\sum_{\substack{i: \text{ left unblocked index}}} \Abs{\depint{i}}
+ \sum_{\substack{i: \text{ right unblocked index}}} \Abs{\depint{i}}\\
\label{eq:variable count}
\sum_{\substack{X \in \variables\text{ left unblocked}}} n_X +
\sum_{\substack{X \in \variables\text{ right unblocked}}} n_X\\
\label{eq:depint count}
\sum_{\substack{i: \text{ left unblocked index}}} 1
+ \sum_{\substack{i: \text{ right unblocked index}}} 1
\end{align}

In all cases the $\Abs{\alpha}$, where $\alpha$ is a variable or
a letter in the original equation $U_0 = V_0$,
denotes the bit-size of encoding size of $\alpha$ in the input equation.
If all letters and variables in the input equation are encoded using bit sequences
of the same length (over some alphabet)
then~\eqref{eq:variable weight} is equivalent to~\eqref{eq:variable count}
and~\eqref{eq:depint weight} to~\eqref{eq:depint count}.

The idea of the steps is: \eqref{eq:variable weight} upper-bounds the increase of bit-size
of encoding depints in the equation after popping letters.
So by iteratively halving it we ensure that total encoding increase caused by popping letters is small.
Similarly, \eqref{eq:depint weight} upper-bounds the increase of size of encodings due to extension of $\suppos(i)$ to new positions.
The following~\eqref{eq:variable count} is connected (in a more complex way) to an increase,
after popping, of number of bits used for numbers in the encoding of letters in $\pos(I)$.
Similarly, \eqref{eq:depint count} to an increase after the extension of depints.

\begin{lemma}
During the pair compression \algsolveeqlin{} can always
choose a partition that at least halves the value
of a chosen non-zero sum among \eqref{eq:variable weight}--\eqref{eq:depint count},
the other sums then do not increase.
\end{lemma}
\begin{proof}
Consider \eqref{eq:variable weight} and take a random partition,
in the sense that each letter $a \in \letters$ goes to the $\letters_\ell$
with probability $1/2$ and to 
$\letters_r$ with probability $1/2$.
Let us fix a variable $X$ and its side, say left.
What happens with $n_X \cdot \Abs{X}$ in~\eqref{eq:variable weight}
in the sum corresponding to left unblocked variables?
If $X$ is left blocked then, by Lemma~\ref{lem:blocking}, it will stay left blocked and so the contribution is and will be $0$.
If it is left unblocked, then its two first letters $a, b$ are not new, so they are in \letters.
If $\sol X$ has only those two letters, then with probability $1/2$ the $a$ will be in $\Gamma_r$
and it will be popped and $X$ will become left blocked
(as \sol X has only one letter),
the same analysis applies, when the third leftmost letter is new.
The remaining case is that the three leftmost letters in \sol X are not new, let them be $a, b, c \in \letters$.
By Lemma~\ref{lem:letters are different} $a \neq b \neq c$.
With probability $1/4$ $ab \in \letters_\ell \letters_r$
and with probability $1/4$ $bc \in \letters_\ell \letters_r$.
Those events are disjoint (as in one $b \in \letters_r$ and in the other $b \in \letters_\ell$)
and so their union happens with probability $1/2$.
In both cases $X$ will become left blocked, as a new letter is its first or second in \sol X.
In all uninvestigated cases the contribution of $n_X \cdot \Abs{X}$
cannot raise.
Also, the analysis for the sum of right-unblocked $X$ is symmetrical.
This shows the claim in this case.

The case of~\eqref{eq:variable count} is shown in the same way as~\eqref{eq:variable weight}.

For~\eqref{eq:depint weight}, the analysis for an index $i$ that is left unblocked is similar,
but this time we consider the positions to the left of $\suppos(i)$
and $\suppos(i)$ can extend to them (instead of letters being popped from variables in case of~\eqref{eq:variable weight})
and some of them may be compressed to one.
Note that if there are no letters to the left/right
then this index is blocked from this side.
The only subtle difference is that only letters in the equation can be in $\suppos(i)$,
so if to the left of $\suppos(i)$ there is a variable,
the $\subpos(i)$ cannot extend to the left.
Thus when the left-most position of $\subpos(I)$ is $p$ and $p-3$ is a new letter,
the previous argument does not work directly.
However, when the letter at position $p-1$ is assigned to $\Gamma_\ell$,
so the case when we would like to argue that it now has $\dep(p-1) \gets \dep(p-1) \cup \dep(p)$ and so $p-1 \in \suppos(I)$,
then this letter is also popped from the variable, so indeed its depint increases
and the whole argument follows. 

The case of~\eqref{eq:depint count} is shown in the same way as~\eqref{eq:depint weight}.
\end{proof}

\subparagraph*{Space consumption}
We now give the linear space bound on the size of equation.
This formalises the intuition from the beginning of Section~\ref{subsubsec:idea}.
As a first step, we show an upper-bound on the encoding size of the equation;
define:
\begin{align*}
	H_d(U, V)
		&=
	\sum_{\substack{i: \text{index}}} \Abs{\depint{i}} \cdot \abs{\suppos(i)}\\
	H_n(U, V)
		&=
	\sum_{\substack{i: \text{index}}} 4 \abs{\suppos(i)} \cdot  \log(\abs{\suppos(i)}+1)\\
	H(U, V)
		&=
	H_d(U, V) + H_n(U, V) \enspace .
\end{align*}
$H_d$ corresponds to bitsize of depints in the encodings
and $H_n$: the numbers (following depints) in the encoding.
\begin{lemma}
	\label{lem:AbsandH}
Given the equation $(U, V)$ it holds that $\Abs{(U, V)} \leq
2\Abs[0]{(U_0, V_0)} + 
2H(U, V)$, where $\Abs[0]{(U_0, V_0)}$ is the bit-size of the encoding of the input equation.
\end{lemma}
\begin{proof}
The variables in the equation are encoded as in the input equation
with the extra prefixed $\texttt{0}$ to distinguished from the letters,
which are prefixed with \texttt{1}.
Thus variables use at most $2\Abs[0]{(U_0, V_0)}$ bits
and it is enough to show that our encoding (without the extra \texttt{1})
uses at most $H(U, V)$ for letters.

A letter at position $p$ is encoded as $\Abs[0]{\depint{\dep(p)}} \#q$,
where $p$ is the $q$-th position in $\pos(\dep(p))$.
We first estimate the space used by
$\Abs[0]{\depint{\dep(p)}}$ and then the one used by $\#q$.
Let $\dep(p) = [i \twodots j]$,
then $\Abs[0]{\depint{\dep(p)}} = \sum_{k=i}^j \Abs[0]{\depint{k}}$
and the space usage is obtained by taking a sum over all positions
(with letters) in the equation $(U, V)$.
When we change the order of grouping and first group by $\depint{k}$,
then summed for $k = 1 \ldots |U_0V_0|$ we obtain
$\sum_{k: \text{ index}} \Abs{\depint{k}} \cdot \abs{\suppos(k)} = H_d(U,V)$.
In numbers:
\begin{align*}
	\sum_{p: \text{ position}} \Abs[0]{\depint{\dep(p)}}
	&=
	\sum_{p: \text{ position}}
	\sum_{k \in \dep(p)} \Abs[0]{\depint{k}}\\
	&=
	\sum_{(p, k): k \in \dep(p)}
	\Abs[0]{\depint{k}}\\
	&=
	\sum_{(p, k): p \in \suppos(k)}
	\Abs[0]{\depint{k}}\\
	&=
	\sum_{k} \Abs[0]{\depint{k}} |\suppos(k)|\\
	&=
	H_d(U, V) \enspace .
\end{align*}

Let us now move to the space usage of numbers in the encoding (so the ones following the depints)
Given a depint $I$ each letter in $\pos(I)$ is assigned a number from
$1$ to $|\pos(I)|$, which is encoded on $\lceil \log(|\pos(I)|+1) \rceil$ bits.
So a number for the position $p$ uses $\lceil \log(|\pos(\dep(p))|+1) \rceil$ bits
and so the space usage for all positions is:
\begin{align*}
\sum_{p: \text{ position}} \lceil \log(|\pos(\dep(p))|+1) \rceil
	&=
\sum_{I: \text{ depint}} |\pos(I)| \lceil \log(|\pos(I)|+1) \rceil
\end{align*}
For each depint $I$ we choose an index $i_I \in I$ such that
one $i$ is chosen at most twice over all depints, this is done as follows:
We know that depints are linearly ordered by $\leq$, see~\Iref{2}.
Fix two consecutive depints in this order $I \leq I'$,
let $I = [i \twodots j]$, $I' = [i' \twodots j']$.
If $i < i'$ then we choose $i_I = i$ and if $i = i'$ then we choose $i_I = j$,
so $i_I$ is one of two ends of $I$.
If $I'$ is the last depint, then if $i = i'$ then we choose $i_{I'} =j'$
and otherwise $i_{I'} = i'$.

Suppose that some fixed $i$ is chosen for $I < I'$ as the beginning.
But this cannot be, as by the choice of $i = i_I$ we have that all following depints
do not include $i$.
So suppose that $i$ was chosen twice for $I < I'$ as the end.
But this cannot be, as by the choice of $i = i_I$ we have that for all following
depints include $i+1$ or some larger index.

As $i_I \in I$ we have $\suppos(i_I) \supseteq \pos(I)$ and so
$|\suppos(i_I)| \geq |\pos(I)|$. Hence
\begin{align*}
\sum_{I: \text{ depint}} |\pos(I)| \lceil \log(|\pos(I)|+1) \rceil
	&\leq
\sum_{I: \text{ depint}} |\suppos(i_I)|\lceil \log(|\suppos(i_I)|+1) \rceil\\
\notag	&\leq
2\sum_{i: \text{ index}} |\suppos(i)| \lceil \log(|\suppos(i)|+1) \rceil\\
\notag	&\leq
4 \sum_{i: \text{ index}} |\suppos(i)| \log(|\suppos(i)|+1)\\
\notag	&=
H_n(U,V)
\qedhere
\end{align*}
\end{proof}

Instead of showing a linear bound on $\Abs{(U , V)}$
we give a linear bound on $H(U, V)$.
Recall that $(U_0, V_0)$ denotes the input equation.

\begin{lemma}
\label{lem:linear space}
	Consider an equation $U = V$, its solution \solution,
a phase of \algsolveeqlin{} which makes the nondeterministic choices
according to \solution{} and partitions
according to the strategy.
Let the returned equation be $(U', V')$.
Then $H (U',V') \leq \frac{5}{6}H(U, V) + \alpha \Abs{(U_0, V_0)}$
and for an intermediate equation $(U'',V'')$ we have
$H(U'',V'') \leq \beta H(U, V) + \gamma \Abs{(U_0, V_0)}$
for some constants $\alpha, \beta, \gamma$.
\end{lemma}
\begin{proof}
We separately estimate the $H_d$ and $H_n$.
Concerning $H_d$, let us first estimate \Abs{ \depint{\dep(p)}} summed over positions $p$ of letters
popped into the equation during a phase
(note, this does not include the size of numbers used in the encoding).
For each variable we pop perhaps several letters to the left and right before block compression,
but those letters are immediately replaced with single letters, so we count each as $1$;
also, when this side of a variable becomes blocked, it can pop at most one letter
(Lemma~\ref{lem:blocking}).
Otherwise, a side of a variable pops at most $1$ letter per pair compression,
in which it is unblocked from this side.
Note that the depint is the same as for the variable,
so the encoding size is \Abs{X}.
So in total the bit-size of popped letters is at most:
\begin{multline}
\label{eq:popping estimation}
\underbrace{\sum_{X \in \mathcal X} 2 n_X \cdot \Abs{X}}_{\text{block compression}}
+ 
\underbrace{\sum_{X \in \mathcal X} 2 n_X \cdot \Abs{X}}_{\text{after $X$ becomes blocked}}
+ \\ +
\sum_{P: \text{ partition}} \left(\sum_{\substack{X \in \variables\\ \text{left unblocked in $P$}}} n_X \cdot \Abs{X} +
\sum_{\substack{X \in \variables\\ \text{right unblocked in $P$}}} n_X \cdot \Abs{X}\right).
\end{multline}
Observe that the third sum (the one summed over all partitions)
at the beginning of the phase is equal to $\sum_{X} 2 n_X \cdot \Abs{X}$,
as no side of the variable is blocked,
and by the strategy point \eqref{eq:variable weight} its value at least halves every 4th pair compression
(and it cannot increase, as by Lemma~\ref{lem:blocking} no side of the variable can cease to be blocked).
Thus~\eqref{eq:popping estimation} is at most
\begin{equation*}
4 \sum_{X} n_X \cdot \Abs{X} + 8 \sum_{X}  n_X \cdot \Abs{X} \left( 1 + \frac{1}{2} + \frac{1}{4} + \cdots \right) = 
20\sum_{X} n_X \cdot \Abs{X} \leq 20\Abs{(U_0, V_0)} \enspace.
\end{equation*}
We now similarly estimate how many positions got into $\suppos(i)$
due to expansion of $\suppos(i)$:
$\suppos(i)$ can expand to two letters during the block compression
(to be more precise: to positions that are inside a block
and to neighbouring blocks to the left/right of the block,
but positions in a block are replaced with a single position and one of them was in $\suppos(i)$,
so there is no increase in the middle block)
to one position at each side after $i$ becomes blocked (Lemma~\ref{lem:blocking})
and by one position for each partition $P$ in which this side of $i$ is not blocked.
So the increase in the bit-size is
\begin{multline}
\label{eq:depint extending estimation}
\underbrace{\sum_{i: \text{ index}} 2 \Abs{\depint{i}}}_{\text{block compression}}
+
\underbrace{\sum_{i: \text{ index}} 2 \Abs{\depint{i}}}_{\text{after blocked}}
+ \\
+
\sum_{P: \text{ partition}}
\Big(\sum_{\substack{i \text{: index}\\\text{left unblocked in $P$ }}} \Abs{\depint{i}}
+
\sum_{\substack{i \text{: index}\\\text{right unblocked in $P$ }}} \Abs{\depint{i}}
\Big)
\end{multline}
and as in \eqref{eq:popping estimation} similarly at the beginning of the phase the third sum
(so the one summed by partitions)
is $\sum_{i: \text{ index}} 2 \Abs{\depint{i}} = 2 \Abs{(U_0,V_0)}$
and it at least halves every 4th partition, by strategy point~\eqref{eq:depint weight}.
Thus similar calculations show that~\eqref{eq:depint extending estimation} is at most $20\Abs{(U_0, V_0)}$.

On the other hand, the number of positions in $\suppos(i)$
drops till the end of the phase by at least $\frac{\abs{\suppos(i)}}{3}-1$
due to compression:
\begin{itemize}
	\item If $\depint{i}$ is a letter, then $\suppos(i)$
	are all positions of letters and Lemma~\ref{lem:shortening the strings}
	yields that $\suppos(i)$ looses at least $\frac{\abs{\suppos(i)}-1}{3}$ positions.
\item If $\depint{i}$ is an ending marker,
	then the marker itself is unchanged and the remaining
	positions in $\suppos(i)$ are letter-positions and Lemma~\ref{lem:shortening the strings}
	applies to them, so $\suppos(i)$ looses at least
	$\frac{\abs{\suppos(i)}-2}{3}<
	\frac{\abs{\suppos(i)}}{3}-1$ positions.
\item If $\depint{i}$ is a variable
	then $\suppos(i)$ includes the position of a variable and
	Lemma~\ref{lem:shortening the strings} applies to strings of letters to the left and right,
	say of length $\ell, r$, where $\ell + r = \abs{\suppos(i)}-1$.
	Then due to compressions $\suppos(i)$ looses at least
	$\frac{\ell - 1}{3} + \frac{r-1}{3} =
	\frac{\abs{\suppos(i)}}{3} -1 $ positions.
\end{itemize}
Thus:
\begin{align*}
H_d(U',V') &\leq
\underbrace{40 \Abs{(U_0, V_0)}}_{\text{new positions in depints}} + \underbrace{\sum_{i\text{: index}} \Abs{\depint{i}} \cdot \left(\frac{2}{3}\abs{\suppos(i)} +1\right)}_{\text{old positions removed}}
\\&=
40 \Abs{(U_0, V_0)} 
+\sum_{i\text{: index}} \frac{2}{3} \Abs{\depint{i}} \cdot \abs{\suppos(i)}
+ \sum_{i\text{: index}} \Abs{\depint{i}}
\\&=
41 \Abs{(U_0, V_0)} +  \frac{2}{3} H_d(U,V)\enspace.
\end{align*}

We also estimate the maximal value of $H_d$
during the phase,
as for intermediate equations we cannot guarantee that
the compression reduced the length of all letters.
We already showed that
in a phase we increase $H_d$ by $40 \Abs{(U_0, V_0)}$.
This yields a~bound of $H_d(U, V) + 40 \Abs{(U_0, V_0)}$,
which shows the part of the claim of Lemma for $H_d$.

Concerning $H_n$, for an index $i$ let $b_i, p_i, e_i$ denote, respectively:
$\abs{\suppos(i)}$ at the beginning of the phase,
number of positions of letters popped from a variable with depint $i$
and number of positions to which $\suppos(i)$ extended due to increase of depints.
To shorten the notation, let $h(x) = x \log(x+1)$.
First we estimate 
$\sum_{i: \text{ index} } h(p_i)$ and 
$\sum_{i: \text{ index} } h(e_i)$
and then use those estimations to calculate the bound on $H_n(U', V')$.
We first inspect the case of $p_i$;
let $P_1, P_2, \ldots$ denote the consecutive partitions in phase.
We show that
\begin{align}
\label{eq:log to square}
\sum_{\substack{i: \text{ index}}} h(p_i) \leq 
\sum_{X\in \variables} 25n_X + \sum_{m\geq 1} m \cdot 
\Big(\sum_{\substack{X \in \variables\\\text{left unblocked in $P_m$}}} n_X 
	+
\sum_{\substack{X \in \variables\\\text{right unblocked in $P_m$}}} n_X\Big) .
\end{align}
To see this consider one occurrence of $X$, let it popped $p_X$ letters.
Then it was not blocked on left/right side for $p_{X,\ell}$/$p_{X,r}$ partitions,
where $p_{X,\ell} + p_{X,r} \geq p_X - 4$,
as from each side we can pop once for block compression
(formally, a sequence is popped, but it is immediately replaced with a single letter),
once after the side becomes blocked and at most once for each pair compression in which the side is not blocked.
Then in right hand side of~\eqref{eq:log to square} the contribution from
one occurrence of $X$ is at least
\begin{align*}
25 + \sum_{i=1}^{p_{X,\ell}} 1 + \sum_{i=1}^{p_{X,r}} 1
	&=
25 + \frac{p_{X,\ell} (p_{X,\ell} + 1) + p_{X,r}(p_{X,r}+1)}{2}\\
	&=
25 + \frac{p_{X,\ell}^2 + p_{X,r}^2}{2} + \frac{p_{X,r} + p_{X,\ell}}{2}\\
	&\geq
25 + \frac{(p_{X,\ell} + p_{X,r})^2}{4} + \frac{p_{X,r} + p_{X,\ell}}{2}\\
	&\geq
\frac{(p_X-4)^2}{4} + \frac{p_X-4}{2} + 25\\ 
	&\geq
p_X \log (p_X+1) \enspace,
\end{align*}
where last inequality can be checked by simple numerical calculation.
Lastly, in~\eqref{eq:log to square} each $p_i$ is equal to an appropriate $p_X$.

The sum in braces on the right hand side of~\eqref{eq:log to square} initially is at most $2\abs{U_0V_0} \leq 2\Abs{(U_0,V_0)}$
and by strategy choice~\eqref{eq:variable count}
it is at least halved every 4th step. So this sum is at most:
\begin{align*}
\sum_{i \geq 0} \underbrace{(16i+10)}_{\text{4 consecutive steps}} \cdot \underbrace{2\Abs{(U_0,V_0)}}_{\text{initial size}} \cdot 
\left(\frac{1}{2}\right)^i
	&=
32\Abs{(U_0,V_0)} \underbrace{\sum_{i \geq 0} i \cdot \left(\frac{1}{2}\right)^i}_{=2}
	+
20\Abs{(U_0,V_0)} \underbrace{\sum_{i \geq 0} \left(\frac{1}{2}\right)^i}_{=2}\\
&= 104 \Abs{(U_0,V_0)}
\end{align*}
and consequently
\begin{equation}
\label{eq:estimationpD}
\sum_{i: \text{ index} } h(p_i) \leq 129 \Abs{(U_0,V_0)} \enspace.
\end{equation}

The analysis for $e_i$ is similar: for a single index $i$
the estimation on the number of positions to which $\suppos(i)$ extends
is the same as the estimation of number of letters popped from an occurrence of a variable, thus 
\begin{equation}
\label{eq:estimationeD}
\sum_{i: \text{ index} } h(e_i) \leq 129 \Abs{(U_0,V_0)} \enspace.
\end{equation}

We now estimate, how many positions were removed from $\suppos(i)$
due to compression,
recall that $b_i$ is the size of $\suppos(i)$ at the beginning of the phase.
Using the same analysis as in the case of $H_d$, from
Lemma~\ref{lem:shortening the strings} it follows that
at least $\frac{b_i}{3}-1$ positions were removed during the phase due to compression .
Thus
\begin{equation}
\label{eq:Dnumbersbitspace}
H_n(U', V') \leq \sum_{i: \text{ index} }
h\left(\frac{2}{3} b_i + 1 + p_i + e_i \right).
\end{equation}
Consider two subcases: if $\frac{2}{3} b_i + 1 + p_i + e_i \leq \frac{5}{6} b_i$
(which implies $b_i \geq 6$),
then the summand can be estimated as $h(\frac{5}{6} b_i) \leq \frac{5}{6} h(b_i)$
and we can upper bound the sum over those cases by 
$\frac{5}{6} \sum_{i: \text{ index} } h(b_i)$.
If $\frac{2}{3} b_i + 1 + p_i + e_i > \frac{5}{6} b_i$
then $1 + p_i + e_i > \frac{1}{6} b_i$ and so
$\frac{2}{3} b_i + 1 + p_i + e_i < 5 (1  + p_i + e_i)$. 
Thus \eqref{eq:Dnumbersbitspace} is upper-bounded by:
\begin{align*}
	H_n(U', V') &< \frac{5}{6} \sum_{i: \text{ index} } h(b_i)
	+
	\sum_{i: \text{ index} } h(5(1 + p_i + e_i)).
	\intertext{In the following, we estimate the second sum.
	As $h$ is convex, we get $h(x+y+z) \leq (h(3x) + h(3y) + h(3z))/3$
	by Jensens' inequality and so}
	\sum_{i: \text{ index} } h(5(1 + p_i + e_i))
		&\leq
	\frac 1 3 \sum_{i: \text{ index} } h(15) + h(15 p_i) + h(15 e_i) \enspace .
	\intertext{
	Consider $h(15 x)$ for natural $x$. If $x= 0 $ then $h(15 x)  = h(x) = 0$
	and othwerwise}
	h(15 x) &= 15 x \log(15 x + 1)\\
		&\leq
	15 x\log(15 (x+1))\\
		&<
	15 x  (4 + \log(x + 1))\\
		&\leq
	60 x\log(x+1)\\
		&=
	60h(x)
	\intertext{And so}
	\frac 1 3 \sum_{i: \text{ index} } h(15) + h(15 p_i) + h(15 e_i)
		&\leq
	\sum_{i: \text{ index} } 20 + 20 h(p_i) + 20h(e_i)\\
		&\leq
	20 \Abs{(U_0,V_0)}  + 2580 \Abs{(U_0,V_0)} + 2580 \Abs{(U_0,V_0)}\\
	&=
	5180 \Abs{(U_0,V_0)} \enspace,
	\intertext{and so plugging into the initial estimations we get}
H_n(U', V') &\leq \frac{5}{6} \sum_{i: \text{ index} } h(b_i)
	+
	5180 \Abs{(U_0,V_0)} \enspace ,
\end{align*}
as claimed.

We should estimate the maximal $H_n$ value
during the phase, as inside a phase we cannot guarantee that letters get compressed,
i.e.\ estimate
$\sum_{i: \text{ index} }
h\left(b_i + p_i + e_i \right)$.
Using similar calculation as in the case of~\eqref{eq:Dnumbersbitspace} and properties of $h$
we obtain:
\begin{align*}
	\sum_{i: \text{ index} }
	h\left(b_i + p_i + e_i \right)
	&\leq
	\frac 1 3 \sum_{i: \text{ index} }
	h(3b_i) + h(3p_i) + h(3e_i)
\intertext{similarly as before for $x = 0$ we have $h(3x) = h(x) = 0$ and for $x \geq 1$:}
	h(3 x) &= 3 x \log(3 x + 1)\\
		&\leq
	3 x\log(3 (x+1))\\
		&<
	3 x  (2 + \log(x + 1))\\
		&\leq
	9 x\log(2(x+1))\\
		&=
	9h(x)
\intertext{and so}	
	\frac 1 3 \sum_{i: \text{ index} }
	h(3b_i) + h(3p_i) + h(3e_i)
	&\leq
	3
	\sum_{i: \text{ index} }
	h(b_i) + h(p_i) + h(e_i)\\
	& \leq
	3 H_n(U,V) + 774 \Abs{(U_0,V_0)}	\enspace. \qedhere
\end{align*}
which shows the claim of the Lemma in the case of $H_n$ and so also in case of $H$.
\end{proof}

\subsection{Proof of Theorem~\ref{thm:linspace}}\label{subsec:wrap-up}
By Lemma~\ref{lem:sound and complete}
all our subprocedures are sound, so we never accept an unsatisfiable equation.

We now analyse the nondeterministic choices that yield termination, completeness and linear space consumption.
Consider an equation $U = V$ at the beginning of the phase,
let \letters be the set of letters in this equation.
If it has a solution $\solution'$,
then it also has a solution \solution over \letters such that \abs{\sol X}=\abs{\solution'(X)} for each variable:
we can replace letters outside \letters
with a fixed letter from \letters.
During the phase we will make nondeterministic choices according to this \solution.

Let the equation obtained at the end of the phase be $U' = V'$ and $\solution'$ be the corresponding solution.
Then $\abs{\solution'(U')} \leq \frac{2\abs{\sol U}+1}{3}$
by Lemma~\ref{lem:shortening the strings}
and we begin the next phase with $\solution'$;
if it has some letters that are not used in the equation,
then we switch to other solution $\solution''$ such that
$\abs{\solution''(U')} \leq \abs{\solution'(U')}$.
Hence we terminate after $\Ocomp(\log N)$ phases,
where $N$ is the length of some solution of the input equation;
it is known that there always is a solution which is at mostly doubly exponential~\cite{PlandowskiSTOC},
so we terminate after exponential number of phases.

Let the algorithm (nondeterministically) choose the partitions according to the strategy.
We show by induction that for an equation $(U, V)$
at the beginning of a phase $H(U,V) \leq \delta \Abs{(U_0, V_0)}$,
where $\delta$ is a constant.
Initially 
$H_n(U_0, V_0) = \Abs{(U_0, V_0)}$ and $H_d(U_0,V_0) = 4\Abs{(U_0, V_0)}$,
as for each index $\dep(i) = \{i\}$;
hence the claim holds.
By Lemma~\ref{lem:linear space} the inequality at the end of each
phase holds for $\delta = 6 \alpha$ for $\alpha$ from Lemma~\ref{lem:linear space}.
For intermediate equations $(U'',V'')$ it holds that 
$H(U'',V'') \leq (6 \alpha \gamma + \beta)\Abs{(U_0, V_0)}$,
by Lemma~\ref{lem:linear space},
where $\beta, \gamma$ are the constants from Lemma~\ref{lem:linear space}.

To upper-bound the space consumption, we also estimate other stored information:
we also store the alphabet from the beginning of the phase
(this is linear in the size of the equation at the beginning of the phase)
and the mapping of this alphabet to the current symbols
(linear in the equation at the beginning of the phase plus the size
of the current equation).
The terminating condition that some pair of letters in $\letters^2$
was not covered is guessed nondeterministically,
we do not store $\letters^2$.
The pair compression and block compression
can be performed in linear space, see Lemma~\ref{lem:performing subprocedures in linear space}.
Note that this includes the change of Huffman coding.


\begin{thebibliography}{10}

\bibitem{hyperbolic_groups_new}
Laura Ciobanu and Murray Elder.
\newblock Solutions sets to systems of equations in hyperbolic groups are
  {EDT0L} in {PSPACE}.
\newblock In Christel Baier, Ioannis Chatzigiannakis, Paola Flocchini, and
  Stefano Leonardi, editors, {\em 46th International Colloquium on Automata,
  Languages, and Programming, {ICALP} 2019, July 9-12, 2019, Patras, Greece},
  volume 132 of {\em LIPIcs}, pages 110:1--110:15. Schloss Dagstuhl -
  Leibniz-Zentrum f{\"{u}}r Informatik, 2019.

\bibitem{twistedwordequations}
Volker Diekert and Murray Elder.
\newblock Solutions of twisted word equations, {EDT0L} languages, and
  context-free groups.
\newblock In Ioannis Chatzigiannakis, Piotr Indyk, Fabian Kuhn, and Anca
  Muscholl, editors, {\em {ICALP}}, volume~80 of {\em LIPIcs}, pages
  96:1--96:14. Schloss Dagstuhl - Leibniz-Zentrum fuer Informatik, 2017.

\bibitem{Diekertfreegroups}
Volker Diekert, Claudio Guti{\'e}rrez, and Christian Hagenah.
\newblock The existential theory of equations with rational constraints in free
  groups is {PSPACE}-complete.
\newblock {\em Inf. Comput.}, 202(2):105--140, 2005.

\bibitem{eqtraces}
Volker Diekert, Artur Je\.z, and Manfred Kufleitner.
\newblock Solutions of word equations over partially commutative structures.
\newblock In Ioannis Chatzigiannakis, Michael Mitzenmacher, Yuval Rabani, and
  Davide Sangiorgi, editors, {\em {ICALP}}, volume~55 of {\em LIPIcs}, pages
  127:1--127:14. Schloss Dagstuhl---Leibniz-Zentrum fuer Informatik, 2016.

\bibitem{wordeqgroups}
Volker Diekert, Artur Je\.z, and Wojciech Plandowski.
\newblock Finding all solutions of equations in free groups and monoids with
  involution.
\newblock {\em Inf. Comput.}, 251:263--286, 2016.

\bibitem{DiekertLohrey08}
Volker Diekert and Markus Lohrey.
\newblock Word equations over graph products.
\newblock {\em International Journal of Algebra and Computation},
  18(3):493--533, 2008.

\bibitem{Gutierrez98}
Claudio Guti{\'e}rrez.
\newblock Satisfiability of word equations with constants is in exponential
  space.
\newblock In {\em FOCS}, pages 112--119, 1998.

\bibitem{Jaffar90}
Joxan Jaffar.
\newblock Minimal and complete word unification.
\newblock {\em J. ACM}, 37(1):47--85, 1990.

\bibitem{wordequations}
Artur Je\.z.
\newblock Recompression: a simple and powerful technique for word equations.
\newblock {\em J.~{ACM}}, 63(1):4:1--4:51, Mar 2016.

\bibitem{contextunification}
Artur Je\.z.
\newblock Deciding context unification.
\newblock {\em J. {ACM}}, 66(6):39:1--39:45, 2019.

\bibitem{Koscielski}
Antoni Ko\'scielski and Leszek Pacholski.
\newblock Complexity of {Makanin}'s algorithm.
\newblock {\em J. ACM}, 43(4):670--684, 1996.

\bibitem{makanin}
Gennadi{\'{\i}} Makanin.
\newblock The problem of solvability of equations in a free semigroup.
\newblock {\em Matematicheskii Sbornik}, 2(103):147--236, 1977.
\newblock (in Russian).

\bibitem{mak82}
Gennadi{\'{\i}} Makanin.
\newblock Equations in a free group.
\newblock {\em Izv. Akad. Nauk SSR}, Ser. Math. 46:1199--1273, 1983.
\newblock English transl. in Math. USSR Izv. 21 (1983).

\bibitem{mat97lfcs}
Yuri Matiyasevich.
\newblock Some decision problems for traces.
\newblock In Sergej Adian and Anil Nerode, editors, {\em LFCS}, volume 1234 of
  {\em LNCS}, pages 248--257. Springer, 1997.
\newblock Invited lecture.

\bibitem{PlandowskiSTOC}
Wojciech Plandowski.
\newblock Satisfiability of word equations with constants is in {NEXPTIME}.
\newblock In {\em STOC}, pages 721--725. {ACM}, 1999.

\bibitem{PlandowskiFOCS}
Wojciech Plandowski.
\newblock Satisfiability of word equations with constants is in {PSPACE}.
\newblock {\em J.~ACM}, 51(3):483--496, 2004.

\bibitem{PlandowskiSTOC2}
Wojciech Plandowski.
\newblock On \emph{PSPACE} generation of a solution set of a word equation and
  its applications.
\newblock {\em Theor. Comput. Sci.}, 792:20--61, 2019.

\bibitem{PlandowskiICALP}
Wojciech Plandowski and Wojciech Rytter.
\newblock Application of {Lempel}-{Ziv} encodings to the solution of word
  equations.
\newblock In Kim~Guldstrand Larsen, Sven Skyum, and Glynn Winskel, editors,
  {\em ICALP}, volume 1443 of {\em LNCS}, pages 731--742. Springer, 1998.

\bibitem{raz87}
Alexander~A. Razborov.
\newblock {\em On Systems of Equations in Free Groups}.
\newblock PhD thesis, Steklov Institute of Mathematics, 1987.
\newblock In Russian.

\bibitem{hyperbolic_groups_old}
Eliyahu Rips and Zlil Sela.
\newblock Canonical representatives and equations in hyperbolic groups.
\newblock {\em Inventiones Mathematicae}, 120:489--512, 1995.

\bibitem{Schulz90}
Klaus~U. Schulz.
\newblock Makanin's algorithm for word equations---two improvements and
  a~generalization.
\newblock In Klaus~U. Schulz, editor, {\em IWWERT}, volume 572 of {\em LNCS},
  pages 85--150. Springer, 1990.

\end{thebibliography}

\end{document}